\documentclass{cjtcs}

\cjtcsdetails{
number=04, year=2020,
received={June 10, 2019},
revised={August 11, 2020},
published={August 30, 2020},
doi={10.4086/cjtcs.2020.004}
 }

\runningtitle{Coin Theorems and the Fourier Expansion}

\cjtcspdftitle{Coin Theorems and the Fourier Expansion}

\runningauthor{Rohit Agrawal}
\copyrightauthor{Rohit Agrawal}
\cjtcspdfauthor{Rohit Agrawal}

\usepackage{mathtools}

\makeatletter
\newcommand{\SwappedDeclarePairedDelimiter}[3]{%
	\expandafter\DeclarePairedDelimiter\csname Swapped\string#1\endcsname{#2}{#3}%
	\newcommand#1{%
		\@ifstar{\csname Swapped\string#1\endcsname}
						{\@ifnextchar[{\csname Swapped\string#1\endcsname}
													{\csname Swapped\string#1\endcsname*}%
						}%
	}%
}
\makeatother

\SwappedDeclarePairedDelimiter{\of}{(}{)}

\providecommand\given{}
\newcommand\SetSymbol[1][]{%
\nonscript\:#1\vert
\allowbreak
\nonscript\:
\mathopen{}}
\DeclarePairedDelimiterX\setd[1]\{\}{%
\renewcommand\given{\SetSymbol[\delimsize]}
#1
}

\def\set{\setd*}

\SwappedDeclarePairedDelimiter{\abs}{\lvert}{\rvert}

\SwappedDeclarePairedDelimiter{\brof}{[}{]}

\SwappedDeclarePairedDelimiter{\ceil}{\lceil}{\rceil}

\DeclareMathOperator*{\Exp}{\mathbb{E}}
\newcommand{\E}[2][]{\Exp_{#1}\brof{#2}}

\newcommand{\PR}[2][]{\Pr_{#1}\brof{#2}}

\newcommand{\eps}{\varepsilon}

\usepackage[capitalize,noabbrev]{cleveref}

\crefformat{footnote}{#2\footnotemark[#1]#3}

\usepackage{tikz}
\usetikzlibrary{shapes.geometric,arrows,positioning}

\newcommand{\RR}{\mathbb{R}}
\newcommand{\F}{\mathcal{F}}
\newcommand{\pmone}{\set{-1, 1}}

\newcommand{\fourier}[2]{\widehat{#1}\of{#2}}
\newcommand{\fourierset}[2]{\fourier{#1}{\set{#2}}}
\newcommand{\levelk}[1][k]{\mathcal{L}_1^{#1}}
\newcommand{\ti}[2][\brof]{\mathbf{I}#1{#2}}
\newcommand{\coinset}[2]{\ensuremath{X_{#1}^{#2}}}
\newcommand{\coins}[1][\eps]{\coinset{#1}{n}}
\newcommand{\fdiff}[1][\eps]{\E{f\of{\coins[#1]}} - \E{f\of{\coins[0]}}}
\newcommand{\abseps}{\abs{\eps}}
\newcommand{\dist}[1]{\mathcal{D}^{(n)}\of{#1}}

\newcommand{\randf}{\pmb{\tilde f}}
\newcommand{\randg}{\pmb{\tilde g}}
\newcommand{\randF}{\pmb{\tilde{\mathcal{F}}}}
\DeclareMathOperator{\Var}{Var}

\newtheorem{thm}{Theorem}[section]
\newtheorem{lem}[thm]{Lemma}
\newtheorem{prop}[thm]{Proposition}
\newtheorem{cor}[thm]{Corollary}

\theoremstyle{definition}
\newtheorem{defn}[thm]{Definition}
\theoremstyle{remark}

\newtheorem{ques}[thm]{Question}

\begin{document}

\begin{frontmatter}

\title{Coin Theorems and the Fourier Expansion}

\author[rohit]{Rohit Agrawal\thanks{Supported by the Department of Defense
	(DoD) through the National Defense Science \& Engineering Graduate
	Fellowship (NDSEG) Program.}}

\begin{abstract}
	In this note we compare two measures of the complexity of a class $\F$
	of Boolean functions studied in (unconditional) pseudorandomness:
	$\F$'s ability to distinguish between biased and uniform coins (the
	\emph{coin problem}), and the norms of the different levels of the
	Fourier expansion of functions in $\F$ (the \emph{Fourier growth}).
	We show that for coins with low bias $\eps = o(1/n)$, a function's
	distinguishing advantage in the coin problem is essentially equivalent
	to $\eps$ times the sum of its level $1$ Fourier coefficients, which
	in particular shows that known level $1$ and total influence bounds
	for some classes of interest (such as constant-width read-once
	branching programs) in fact follow as a black-box from the
	corresponding coin theorems, thereby simplifying the proofs of some
	known results in the literature. For higher levels, it is well-known
	that Fourier growth bounds on all levels of the Fourier spectrum imply
	coin theorems, even for large $\eps$, and we discuss here the
	possibility of a converse.
\end{abstract}


\end{frontmatter}


\section{Introduction}

A natural question one can ask when studying a limited model of
computation is how well it can solve some basic computational
task, such as distinguishing between an unfair and a fair coin:
\begin{defn}[Coin problem]
	The \emph{advantage $\alpha$} of a Boolean function $f:\pmone^n \to \pmone$ in
	\emph{distinguishing $\eps$-biased coins from uniform coins} is
	defined to be
	\[
		\alpha = \abs{\fdiff},
	\]
	where $\coins[\delta]$ are iid random variables over $\pmone$ with
	expectation $\delta$, so that $\coins[0]$ are uniform random bits.

	A set $\F$ of Boolean functions is said to \emph{solve the $\eps$-coin
	problem with advantage $\alpha$} for
	\[
		\alpha = \max_{f\in \F} \abs{\fdiff}.
	\]
\end{defn}

One can see (e.g.~via the equivalence of $\ell_1$ and total variation
distance of probability distributions) that the unique Boolean function
$f$ achieving the greatest distinguishing advantage is of the form $f(x)
= 1$ if and only if the number of $1$s in $x$ is at least $k$, for the
smallest $k$ such that $(1 + \eps)^k(1-\eps)^{n-k} \geq 1$. This is
simply a (symmetric) linear threshold function, and so one sees that the
coin problem for a class $\F$ of Boolean functions is closely related to
the ability of $\F$ to approximate threshold functions, and in
particular the majority function. The study of threshold functions in
limited models is extensive, with early works due to e.g.~Ajtai
\cite{ajtai_sigma_1983} and Valiant \cite{valiant_short_1984}, and early
explicit consideration of the coin problem by Shaltiel and Viola
\cite{sha_vio_hardness_2010}, Brody and Verbin (who introduced the name)
\cite{bro_ver_coin_2010}, and Aaronson \cite{aaronson_bqp_2010}. These
results are generally stated as a \emph{coin theorem}, giving a
statement of the form $\F_n$ (parametrized by some parameter, e.g.~the
input length) cannot solve the $\eps$-coin problem except with some
advantage $\beta(n, \eps)$. Subsequent work has given tight coin
theorems on constant-width read-once branching programs (Steinberger
\cite{steinberger_distinguishability_2013}), $\mathsf{AC^0}$ (Cohen,
Ganor, and Raz \cite{coh_gan_raz_two_2014}), $\mathsf{AC^0}[\oplus]$
(Limaye, Sreenivasaiah, Srinivasan, Tripathi, and Venkitesh
\cite{lim_sre_sri_tri_ven_fixeddepth_2018}), and product tests (Lee and
Viola \cite{lee_vio_coin_2018}).

The other main object we will discuss is the Fourier expansion of $f$,
for which we follow the notation of O'Donnell
\cite{odonnell_analysis_2014}.

\begin{defn}\label{defn:fourier_expansion}
	The \emph{Fourier expansion} of a function $f:\pmone^n \to \RR$ is its
	unique representation as a multilinear polynomial
	\[
		f(x) = f(x_1, \dots, x_n) = \sum_{S\subseteq [n]} \fourier{f}{S}
		\prod_{i\in S} x_i\,,
	\]
	where $[n] \coloneqq \set{1,
	2, \dotsc, n}$ and $\fourier{f}{S}$ is called the \emph{Fourier
	coefficient of $f$ on $S$}. The coefficient $\fourier{f}{S}$ is said
	to be at \emph{level $\abs{S}$}, and can be expressed as
	\[
		\fourier{f}{S} = \E[x]{f(x) \prod_{i\in S} x_i}
	\]
	where the expectation is taken over the uniform distribution over
	$\pmone^n$. In particular, $\fourier{f}{\emptyset} = \E[x]{f(x)}$.

	The \emph{total influence of $f$}, denoted $\ti{f}$, is defined as
	$\ti{f} = \sum_{i=1}^n \sum_{S\ni i} \fourier{f}{S}^2 =
	\sum_{S\subseteq [n]}\abs{S} \fourier{f}{S}^2$. If $f:\pmone^n \to
	\pmone$ is monotone, then the total influence has the simpler
	expression $\ti{f} = \sum_{i=1}^n \fourierset{f}{i}$.
\end{defn}

The use of Fourier analysis in the study of Boolean functions was
pioneered by the work of Kahn, Kalai, and Linial
\cite{kah_kal_lin_influence_1988}, which kicked off a long line of work,
including early results analyzing DNFs and CNFs by Brandman, Orlitsky, and
Hennessy \cite{bra_orl_hen_spectral_1990}, and $\mathsf{AC^0}$ by
Linial, Mansour, and Nisan \cite{lin_man_nis_constant_1993}.  More
recently, Reingold, Steinke, and Vadhan
\cite{rei_ste_vad_pseudorandomness_2013} introducted the concept of
\emph{Fourier growth} for use in constructions of pseudorandom
generators via the iterative random restriction framework of Ajtai and
Wigderson \cite{ajt_wig_deterministic_1989}, and particularly in the
sense of the later work of Gopalan, Meka, Reingold, Trevisan, and Vadhan
\cite{gop_mek_rei_tre_vad_better_2012}.

\begin{defn}[\cite{rei_ste_vad_pseudorandomness_2013}]
	Given $f:\pmone^n \to \mathbb R$, we define its \emph{level $k$ $\ell_1$ norm} to be
	\[
		\levelk(f) = \sum_{\abs{S} = k} \abs{\fourier{f}{S}}.
	\]
	For a class $\F$ of functions, we define the \emph{level $k$ $\ell_1$
	norm of $\F$} to be $\levelk(\F) = \sup_{f\in \F} \levelk(f)$.
	$\F$ is said to have \emph{Fourier growth with base $t$} if
	$\levelk(\F) \leq t^k$ for all $k$.
\end{defn}

Fourier growth bounds were introduced by
\cite{rei_ste_vad_pseudorandomness_2013} to capture the fact that
(roughly speaking) a random restriction that keeps input bits alive with
probability $p$ reduces the level $k$ $\ell_1$ norm of a function by
$p^k$, and so a random restriction keeping an $O(1/t)$ fraction of bits
alive simplifies a function with Fourier growth with base $t$ into one
with constant total Fourier $\ell_1$ norm, which is then easily
fooled by a small-bias generator (Naor and Naor
\cite{nao_nao_smallbias_1993}).  Fourier growth bounds have been studied
for classes such as $\mathsf{AC^0}$
\cite{lin_man_nis_constant_1993,hastad_slight_2001,tal_tight_2017},
regular and constant-width read-once branching programs
\cite{rei_ste_vad_pseudorandomness_2013,cha_hat_rei_tal_improved_2018},
and product tests \cite{lee_fourier_2019}.

In this work, we study implications between bounds on the advantage of a
class of Boolean functions in the coin problem and bounds on the Fourier
spectrum. Some connections along this line are well known, for example
it is folklore (see e.g.~Tal \cite{tal_tight_2017}) that Fourier growth
bounds of the form $\levelk(f) \leq t^k$ imply a strong coin theorem
with bound $\beta(n,\eps) = \eps\cdot O(t)$ for all $\eps = O(1/t)$.
Recently, Tal \cite{tal_private_2019} (who graciously allowed the
author to include his result in this note) showed that if a class of
functions $\F$ is closed under restriction, then for $\F$ to satisfy
such a strong coin theorem it is enough for it to merely have
$\levelk[1](\F) \leq t$ rather than have Fourier growth with base
$t$.

In this note, we give the first (to the best of our knowledge) results
in the other direction, showing that coin theorems for small $\eps =
o(1/n)$ are \emph{equivalent} to bounds on the sum of the level $1$
Fourier coefficients, and in particular to bounds on $\levelk[1](\F)$
assuming that these coefficients can be taken to be non-negative (e.g.
if $\F$ consists of monotone functions, or is closed under negating
input variables). Thus, perhaps surprisingly, even coin theorems for
only a small range of $\eps$ are sufficient to bound the level $1$
spectrum of the class $\F$, even for $\F$ not closed under restriction.
This allows us to give simple ``black-box'' proofs of existing level $1$
bounds in the literature, such as the bound for constant-width read-once
branching programs due to Steinke, Vadhan, and Wan
\cite[\S6]{ste_vad_wan_pseudorandomness_2017} in 2014, which was conjectured
by \cite{rei_ste_vad_pseudorandomness_2013}.

These results are summarized in \cref{fig:results}, showing the
implications between bounds on coin problems and Fourier growth for a
class $\F$, arranged from top to bottom in order of strength (under the
assumption that $\F$ is monotone or closed under negation of input
variables, but need not be closed under restriction). If $\F$ is
additionally closed under restriction, then the bottom two layers
collapse.

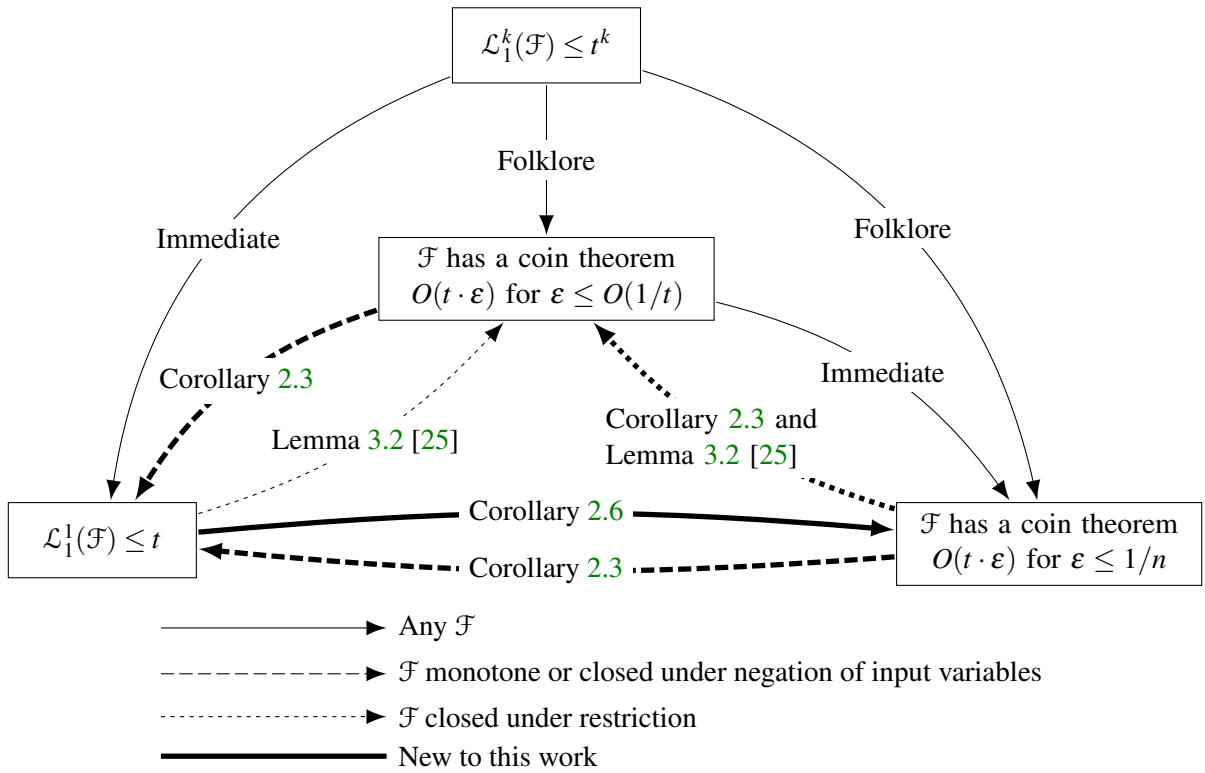
\begin{figure}[!htbp]
	\centering
	\begin{tikzpicture}
		\tikzset{>=latex,ultra thick}
		\tikzstyle{pred} = [rectangle, minimum width=2.5cm, minimum height=1cm, text centered, draw=black];
		\tikzstyle{restriction} = [dotted,draw,->];
		\tikzstyle{negation} = [dash pattern=on 5pt off 2pt,draw,->];
		\tikzstyle{always} = [draw,->];
		\tikzstyle{new} = [line width=2pt];
		\node (lk) [pred] {$\mathcal{L}_1^k(\F) \leq t^k$};
		\node (coint) [pred, text width=4.2cm, below=2cm of lk] {$\F$ has a coin theorem $O(t\cdot \eps)$ for $\eps \leq O(1/t)$};
		\node (l1) [pred,below left=2.37cm and 2.37cm of coint] {$\mathcal{L}_1^1(\F) \leq t$};
		\node (coin1) [pred, text width=3.8cm, below right=2.37cm and 2.37cm of coint] {$\F$ has a coin theorem $O(t\cdot \eps)$ for $\eps \leq 1/n$};

		\path [always] (coint) edge [bend left=20] node[fill=white] {Immediate} (coin1);
		\path [always] (lk) edge[bend right] node[fill=white] {Immediate} (l1);
		\path [always] (lk) -- node[fill=white] {Folklore} (coint);
		\path [always] (lk) edge [bend left] node[fill=white] {Folklore} (coin1);
		\path [restriction] (l1) edge [bend right=15] node[fill=white] {\cref{lem:restrictfl1givescoin} \cite{tal_private_2019}} (coint);
		\path [restriction,new] (coin1) edge [bend left=15] node[fill=white,text width=3.3cm] {{\cref{cor:coin_class} and \cref{lem:restrictfl1givescoin} \cite{tal_private_2019}}} (coint);
		\path [always,new] (l1) edge [bend left=5] node[fill=white] {{\cref{cor:l1coinsmalleps}}} (coin1);
		\path [negation,new] (coin1) edge [bend left=5] node[fill=white] {{\cref{cor:coin_class}}} (l1);
		\path [negation,new] (coint) edge [bend right=20] node[fill=white] {{\cref{cor:coin_class}}} (l1);

		\matrix[below]  at (current bounding box.south) {
			\node[right] (alw) at (3,0) {Any $\F$};
			\path [always] (0, 0) --  (3, 0); \\
			\node[right] (neg) at (3,0) {$\F$ monotone or closed under negation of input variables};
			\path [negation] (0, 0) --  (3, 0); \\
			\node[right] (rest) at (3,0) {$\F$ closed under restriction};
			\path [restriction] (0, 0) --  (3, 0); \\
			\node [right] (thiswork) at (3, 0) {New to this work};
			\path [draw,new,-] (0, 0) -- (3,0); \\
		};
		\end{tikzpicture}
	\caption{Implications studied in this work for a class $\F$ of Boolean functions,
	assuming $t\geq 1$}
	\label{fig:results}
\end{figure}

We also show (non-constructively) the existence of a class of Boolean
functions $\F$ closed under restriction and negations which satisfies a
coin theorem of the form $\eps \cdot O(\log n)$ for all $\eps \leq
O(1/\log n)$ (and thus by \cref{cor:coin_class} has $\levelk[1](\F) \leq
O(\log n)$ as shown in \cref{fig:results}), but yet has $\levelk[3](\F)
= \Omega(\log n \cdot n)$.  This result may be of interest because,
intriguingly, many natural classes of Boolean functions $\F$ with known
coin theorems or level-$1$ bounds are also known (or conjectured) to
satisfy corresponding Fourier growth bounds (and in fact we are unaware
of any natural class of Boolean functions satisfying these closure
conditions and a $\levelk[1](\F)\leq t$ bound or coin theorem $\eps
\cdot O(t)$ which does not at least conjecturally also have Fourier
growth with base $O(t)$).  We hope that this result may help point the
way to giving some additional constraints under which one could hope for
a $\levelk[1]$ or coin problem bound to imply a general $\levelk$ bound.

\section{Low bias and level 1}

The main technical result of this section is the following simple
proposition.
\begin{prop}\label{prop:fourier_sum_bound}
	For every $f:\pmone^n \to \mathbb R$, the rate of convergence of the limit
	\[
		\lim_{\eps\to0} \frac{1}{\eps}\of[\Big]{\fdiff} = \sum_{i=1}^n
		\fourierset{f}{i}
	\]
	can be bounded for every $\abs{\eps} \leq 1/\sqrt{n}$ as
	\[
		\abs{\frac{1}{\eps}\of[\Big]{\fdiff} - \sum_{i=1}^n
		\fourierset{f}{i}} \leq \abseps\cdot n \cdot
		\sqrt{\Var\of{f\of{\coins[0]}}}
		\,,
	\]
	where in particular if $f:\pmone^n \to [-1, 1]$ then the right-hand
	side is bounded by $\abseps \cdot n$.
\end{prop}

In particular, bounds on the coin problem for $f$ for small $\eps$
(e.g.~$\eps = o(1/n)$) are equivalent to bounds on the sum of the level
$1$ Fourier coefficients of $f$. We give the straightforward proof of (a
slight strengthening of) \cref{prop:fourier_sum_bound} at the end of
this section, but we first note some immediate corollaries, including
some simple new proofs of known results in the literature on Fourier
growth.

Most basically, if we can assume that the level $1$ Fourier coefficients
of $f$ are non-negative, we get upper bounds on the $\ell_1$ norm of
these coefficients:
\begin{cor}\label{cor:nonnegative_fn}
	Let $f:\pmone^n\to [-1,1]$ be such that $\fourierset{f}{i} \geq 0$ for
	every $1\leq i \leq n$ (e.g.~$f$ is monotone). Then
	\begin{align*}
		\levelk[1]\of{f} =
		\sum_{i=1}^n \abs{\fourierset{f}{i}} \leq \inf_{\abs{\eps} \leq
			1/\sqrt{n}}\set{ \abs{\frac{1}{\eps}\of[\Big]{\fdiff}}  + \abs{\eps} n}
	\end{align*}
	In particular, if $f$ is monotone, then its total influence $\ti{f} =
	\levelk[1]\of{f}$ satisfies the same bound.
\end{cor}

Beyond monotonicity, another way to establish such a bound is if $f$ is
part of a larger class of functions with basic closure properties that
all simultaneously satisfy a coin theorem.
\begin{cor}\label{cor:coin_class}
	Let $\F$ be a class of functions $f:\pmone^n \to [-1, 1]$ such that
	\begin{enumerate}
		\item $\F$ satisfies a coin theorem with bound $\beta(n, \eps)$, meaning
			$\abs{\fdiff} \leq \beta(n, \eps)$ for every $f\in \F$.
		\item For every $f\in \F$ there exists $g\in \F$ such that
			$\fourierset{g}{i} = \abs{\fourierset{f}{i}}$ for every $1\leq i
			\leq n$ (e.g.~$\F$ is closed under negation of input variables, or
			consists only of monotone functions).
	\end{enumerate}
	Then
	\[
		\levelk[1]\of{\F} =
		\sup_{f\in \F}\sum_{i=1}^n \abs{\fourierset{f}{i}} \leq \inf_{\abs{\eps} \leq
			1/\sqrt{n}}\set{ \abs{\frac{\beta(n,\eps)}{\eps}} + \abs{\eps} n}
	\]
	and in particular, every monotone $f\in \F$ has $\ti{f} =
	\levelk[1]\of{f}$ at most the above.
\end{cor}

Note that taking $\eps = 1/n$ in \cref{cor:coin_class} gives the
simple upper bound of $n \cdot \abs{\beta(n, 1/n)} + 1$. To show the
applicability of this result, we show how it can be used to give simple
proofs of several existing results in the literature. As it is not
relevant for us, we will not give formal definitions of the classes $\F$
involved and defer to the original papers for such details.

For constant-width read-once branching programs, Brody and Verbin
\cite{bro_ver_coin_2010} first claimed a coin theorem, which was
improved by Steinberger \cite{steinberger_distinguishability_2013} to
give an optimal bound. Interestingly, Steinberger also separately proved
(using the same techniques) a total influence bound on monotone
constant-width read-once branching programs, which was generalized
(again using the same techniques) by Steinke, Vadhan, and Wan
\cite[\S6]{ste_vad_wan_pseudorandomness_2017} in 2014 to a corresponding
level $1$ Fourier bound for (non-necessarily monotone) constant-width
read-once branching programs. Applying Corollary \ref{cor:coin_class},
we see that these latter results can in fact be derived (up to
constant factors) using Steinberger's coin theorem as a black box.

\begin{cor}[\cite{steinberger_distinguishability_2013,ste_vad_wan_pseudorandomness_2017}]\label{coin:robp}
	Let $f:\pmone^n \to \pmone$ be computable by a width-$w$ read-once
	branching program. Then $\levelk[1](f) \leq O_w(\log n)^{w-2} + O(1)$.
	In particular, if $f$ is monotone then $f$ has total influence $\ti{f}
	= \levelk[1](f)$ satisfying the same bound.
\end{cor}
\begin{proof}
	Steinberger's full coin theorem \cite[Full version, Corollary
	1]{steinberger_distinguishability_2013} states that for every integer
	$r\geq 1$ that
	\[
		\abs{\fdiff} \leq \eps r^{w-2} + \of{n + r^{w-2}}\of{w-2}\of{\frac 1{2-\eps}}^{r-1}
	\]
	so since width-$w$ read-once branching programs are closed under
	negation of input variables, \cref{cor:coin_class} implies that for $0
	< \eps \leq 1/\sqrt{n}$ that
	\[
		\levelk[1]\of{f} \leq r^{w-2} + \frac{1}{\eps}\of{n + r^{w-2}}\of{w-2}\of{\frac 1{2-\eps}}^{r-1}
		 + \eps n
	\]
	so setting $\eps = 1/n$ and taking $r = \ceil{4\log n} + 1$ gives the
	result.
\end{proof}

One thing to note about the above proof is that the coin theorem we used
of \cite{steinberger_distinguishability_2013} was suboptimal in the
range $\eps = n^{-\omega(1)}$, but still implied the optimal level $1$
bound of \cite{ste_vad_wan_pseudorandomness_2017}.  Using
\cref{prop:fourier_sum_bound}, we can also improve the
\cite{steinberger_distinguishability_2013} coin theorem for small $\eps$
by setting $\eps_0 = 1/n$ in the following corollary:
\begin{cor}
	\label{cor:improvecoinverysmall}
	Let $f:\pmone^n \to [-1, 1]$ satisfy a coin theorem of $\abseps\cdot
	B$ for all $\abseps \geq \eps_0$, where $B\geq 0$ and $\eps_0 \leq
	1/\sqrt{n}$.  Then for all $\abseps < \eps_0$, it holds that
	\[
		\abs{\fdiff} \leq \abseps\cdot \of[\big]{B + n(\abseps + \eps_0)}
		\leq \abseps\cdot \of{B + 2n\cdot \eps_0}.
	\]
	In particular, $f$ satisfies a coin theorem with bound $\abseps \cdot
	(B+2n\cdot \eps_0)$ for all $\abseps \leq 1$.
\end{cor}

The proof of \cref{cor:improvecoinverysmall} goes by using
\cref{prop:fourier_sum_bound} applied with $\eps_0$ to derive a bound on
the sum of the level $1$ Fourier coefficients of $f$, then applying the
following converse of \cref{cor:coin_class} to derive a coin theorem
for small $\eps$ from such a bound:
\begin{cor}\label{cor:l1coinsmalleps}
	Let $f:\pmone^n \to [-1,1]$ have $\abs{\sum_{i=1}^n \fourierset fi
	}\leq t$ (e.g.~$\levelk[1](f) \leq t$).  Then for all $\abseps \leq
	1/\sqrt{n}$,
	\[
		\abs{\fdiff} \leq \abseps\cdot\of{t + \abseps \cdot n},
	\]
	so that in particular $f$ satisfies a coin theorem with bound $2t\cdot
	\abseps$ for all $\abseps \leq t/n$.
\end{cor}
\begin{proof}
	The triangle inequality and \cref{prop:fourier_sum_bound} imply that
	\begin{align*}
		\abs*{\fdiff[\eps]}
					&= \abseps \cdot \abs{\frac{1}{\eps}\cdot \of{\fdiff}}\\
		&\leq
		\abs{\eps}\cdot \of{ \abs{\frac{1}{\eps}\of{\fdiff[\eps]}-\sum_{i=1}^n\fourierset{f}{i}}
		+ \abs{\sum_{i=1}^n\fourierset{f}{i}}}
			\\
		&\leq \abs{\eps}\cdot \of{\abs{\eps}\cdot  n +  t}
	\end{align*}
	as desired.
\end{proof}
\begin{proof}[Proof of \cref{cor:improvecoinverysmall}]
	By the triangle inequality and \cref{prop:fourier_sum_bound} we have
	\begin{align*}
		\abs{\sum_{i=1}^n \fourierset fi}&\leq \abs{\sum_{i=1}^n \fourierset fi - \frac1{\eps_0}\of{\fdiff[\eps_0]}} + \abs{\frac{1}{\eps_0}\of{\fdiff[\eps_0]}} \\&\leq \abs{\eps_0}\cdot n + B,
	\end{align*}
	so the result follows by applying \cref{cor:l1coinsmalleps} with $t = B + \abs{\eps_0}\cdot n$.
\end{proof}

It remains to prove \cref{prop:fourier_sum_bound}, for which we will
need just one simple fact about the Fourier expansion.
\begin{lem}[Parseval's identity]
	The Fourier coefficients of $f:\pmone^n \to \RR$ satisfy
	\[
		\sum_{S\subseteq [n]} \fourier{f}{S}^2 = \E{f\of{\coins[0]}^2}.
	\]
	In particular, $\Var\of{f\of{\coins[0]}} = \sum_{\abs S \geq
	1}\fourier{f}{S}^2$.
\end{lem}

\begin{prop}[Strengthened version of \cref{prop:fourier_sum_bound}]
\label{prop:fourier_sum_strengthened}
	For every $f:\pmone^n \to \mathbb R$, we have for every
	$\abs\eps\leq 1/\sqrt n$ that
	\[
		\abs{\frac{1}{\eps}\of[\Big]{\fdiff} - \sum_{i=1}^n
		\fourierset{f}{i}} \leq \abseps\cdot n \cdot
		\sqrt{\sum_{\abs S\geq 2}\fourier fS^2}
		\leq\abs\eps\cdot n\cdot \sqrt{\Var\of{f\of{\coins[0]}}}
		\,.
	\]
	Furthermore, for every $\abs{\eps}\leq 1/(4\sqrt n)$ we
	have the stronger bound that
	\[
		\abs{\frac{1}{\eps}\of[\Big]{\fdiff} - \sum_{i=1}^n
		\fourierset{f}{i}} \leq \abseps\cdot n \cdot
		\sqrt{\max_{k\geq 2}\sum_{\abs S=k}\fourier fS^2}
		\,.
	\]
\end{prop}
\begin{proof}
	Let $f:\pmone^n \to \mathbb R$. Then by the multilinearity of the
	Fourier expansion, we have
	\begin{align*}
		\fdiff
		&= \sum_{S\subseteq [n]}\fourier{f}{S} \of{\E[x\sim
		{\coins}]{\prod_{i\in S} x_i} - \E[x\sim {\coins[0]}]{\prod_{i\in
		S}x_i}}
		= \sum_{\emptyset \neq S \subseteq [n]} \fourier{f}{S} \cdot
		\eps^{\abs{S}}\,.
	\end{align*}
	Rearranging gives
	\[
		\frac 1\eps \of[\Big]{\fdiff} - \sum_{i=1}^n \fourierset{f}{i}
		=	\frac{1}{\eps}\sum_{k=2}^n \eps^k \sum_{\abs{S} = k}\fourier{f}{S}\,,
	\]
	so by the triangle inequality and Cauchy--Schwarz we get
	\begin{align*}
		\abs{\frac 1\eps \of[\Big]{\fdiff} - \sum_{i=1}^n \fourierset{f}{i}}
		&\leq \frac 1\abseps\sum_{k=2}^n \abseps^k \sqrt{\binom{n}{k}\sum_{\abs{S} = k}\fourier{f}{S}^2}
		\,.
	\end{align*}
	We show two different techniques to bound this sum, in terms of either
	$\sum_{\abs S \geq 2}\fourier{f}{S}^2\leq \Var\of{f\of{\coins[0]}}$ or
	$\max_{k\geq 2} \sum_{\abs S = k}\fourier{f}{S}^2$. For the former, we have
	\begin{align*}
		\frac 1\abseps\sum_{k=2}^n \abseps^k&\sqrt{\binom{n}{k}\sum_{\abs{S} = k}\fourier{f}{S}^2}\\
		&=\frac 1\abseps\sum_{k=2}^n \of{\abseps^k\sqrt{\binom{n}{k}}}\cdot\of{\sqrt{\sum_{\abs{S} = k}\fourier{f}{S}^2}}\\
		&\leq \frac1\abseps \sqrt{\sum_{k=2}^n \binom{n}{k}\cdot \eps^{2k}}\cdot \sqrt{\sum_{\abs S \geq 2} \fourier fS^2}
		&&\text{(by Cauchy-Schwarz)}
	\end{align*}
	We conclude by noting that
	$\sum_{k=2}^n \binom{n}{k}\cdot \eps^{2k}=(1 + \eps^2)^n - (1 + n\eps^2) \leq e^{n\eps^2} - (1 + n\eps^2)
	\leq (n\eps^2)^2$ for $n\eps^2 \leq 1$.

	For the other bound, we have since $\binom nk\leq n^k/2$ for $k\geq 2$ that
	\[
			\frac 1\abseps\sum_{k=2}^n \abseps^k \sqrt{\binom{n}{k}\sum_{\abs{S} = k}\fourier{f}{S}^2}
			\leq \frac 1{\sqrt 2\cdot \abseps}\cdot \sum_{k=2}^n \of{\abseps\sqrt{n}}^k\cdot
			\sqrt{\max_{k\geq 2}\sum_{\abs S = k}\fourier{f}{S}^2}
	\]
	where 
	\[
		\frac 1{\sqrt 2\cdot \abseps}\cdot \sum_{k=2}^n \of{\abseps\sqrt{n}}^k
			= \frac 1{\sqrt 2\cdot \abseps} \cdot \frac{\abseps^2\cdot n - \abseps^{n+1}\sqrt{n}^{n+1}}{1 - \abseps\sqrt{n}}
			\leq \frac{\abseps\cdot  n}{\sqrt 2\cdot(1 - \abseps\sqrt{n})}
	\]
	is at most $\abseps\cdot n$ when $\abseps\leq (1-1/\sqrt 2)/\sqrt n<1/(4\sqrt n)$
	as desired.
\end{proof}

\section{Larger bias and beyond Level 1}

The previous section demonstrated that bounds on the level 1 Fourier
coefficients are essentially equivalent to coin theorems for
inverse-polynomially small error $\eps = o(1/n)$. This raises
two natural questions: can we say anything about either coin theorems
for larger $\eps$, or about bounds on the Fourier spectrum beyond
level $1$?

These questions are of interest because, perhaps surprisingly, many
natural classes of Boolean functions $\F$ for which we know level-$1$
bounds are also known (or conjectured) to satisfy corresponding Fourier
growth bounds $\levelk(\F) \leq O\of{\levelk[1](\F)}^k$ for all $k$. For
example, $\mathsf{AC^0}$ (Tal \cite{tal_tight_2017}) and the class of
product tests (Lee \cite{lee_fourier_2019}) are known to have this
property, and constant-width read-once branching programs (cwROBPs) are
believed to (Chattopadhyay, Hatami, Reingold, and Tal
\cite{cha_hat_rei_tal_improved_2018} explicitly make this conjecture and
prove almost this optimal result). Furthermore, these classes all have
known corresponding coin theorems which are not only capable of proving
the known $\levelk[1]$ bound via \cref{cor:coin_class}, but are also
valid for all $\eps = O(1/\levelk[1](\F))$ (Cohen, Ganor, and Raz
\cite{coh_gan_raz_two_2014} for $\mathsf{AC^0}$, Lee and Viola
\cite{lee_vio_coin_2018} for product tests, and Steinberger
\cite{steinberger_distinguishability_2013} for cwROBPs). One might
therefore hope that there is a stronger relationship between Fourier
growth and coin theorems.

One direction, that Fourier growth bounds of this form imply coin
theorems, is well-known (see e.g.~\cite{tal_tight_2017}, this can also
be seen in the proof of \cref{prop:fourier_sum_bound} by replacing the
first Cauchy-Schwarz step), so the goal of this section is to explore
the possibility of a converse. Generally, allowing for both additive and
multiplicative losses, one might ask something like the following:
\begin{ques}\label{conj:meta}
	Is there a ``natural'' set of conditions $\mathcal C$ such that the
	following is true:
	Let $\F = \of{\F_n}_{n\in \mathbb N}$ be a class of Boolean functions
	$f_n:\pmone^n\to \pmone$ satisfying $\mathcal C$.
	Then if $\F$ satisfies a coin theorem with bound
	$\abseps\cdot B(n)$, meaning for all $n\in \mathbb N$, $f_n\in \F_n$
	and $\abseps \leq 1$ it holds that $\abs{\E{f_n\of{\coins[\eps]}} -
	\E{f_n\of{\coins[0]}}} \leq \abseps \cdot B(n)$, then
	there exists a constant $c_{\F}$ such that
	\[
		\levelk(\F) \leq \of{c_{\F}\cdot \of{1 + B(n)}}^k
	\]
	for all $k$.
\end{ques}

Perhaps the most natural condition to impose, beyond closure under
negations as considered in the previous section, is closure under
\emph{restriction}, that is, fixing parts of the input, since
intuitively this has the property of ``reducing the level'' of any
Fourier coefficient containing one of the fixed inputs. Furthermore, all
the classes of Boolean functions mentioned earlier in this section are
closed under restriction, and we are not aware of any natural class of
Boolean functions with these properties which does not (at least
conjecturally) satisfy a corresponding $\levelk$ bound.  However, Tal
\cite{tal_private_2019} recently gave evidence suggesting that this is
not enough, showing that any $\F$ closed under restriction satisfies a
coin theorem of $\eps\cdot O(\levelk[1](\F))$ for $\eps =
O(1/\levelk[1](\F))$, so that for such $\F$ there is essentially no
difference between coin theorems for $\eps = o(1/n)$ and $\eps =
O(1/\levelk[1](\F))$.

\begin{lem}[\cite{tal_private_2019}\footnote{\label{foot:thanks}We thank Avishay Tal for telling us about this result and
	allowing us to include it and its proof in this note.}]\label{lem:restrictfl1givescoin}
	Let $\F$ be a class of Boolean functions which is closed under
	restriction and satisfies $\abs{\sum_{i=1}^n \fourierset{f}{i}} \leq t$
	for every $f\in \F$ (e.g.~$\levelk[1](\F) \leq t$). Then for all
	$\abseps < 1$ it holds that
	\[
		\abs{\fdiff} \leq \ln\of{\frac1{1-\eps}} \cdot t.
	\]
	In particular, $\F$ satisfies a coin theorem of $\abseps \cdot O(t)$
	for all $\abseps = \min(1/t,0.99)$.
\end{lem}
As an example of the power of this result, note that it implies the coin
theorem for $\mathsf{AC^0}$ of Cohen, Ganor, and Raz
\cite{coh_gan_raz_two_2014}  as a corollary of Boppana's
\cite{boppana_average_1997} bound on the total influence of that class.
\begin{cor}[\cite{coh_gan_raz_two_2014}]
	Let $f:\pmone^n \to \pmone$ be computable by a size $s$, depth $d$
	Boolean circuit and $\abseps \leq 1$. Then $\abs{\fdiff} \leq \abseps
	\cdot O_d(\log^{d-1}(s))$.
\end{cor}

We include the proof\cref{foot:thanks} of
\cref{lem:restrictfl1givescoin} at the end of this section, but we will
first use Tal's result to provide a simple proof that there exist
classes of Boolean functions $\F$ which are closed under restriction and
negations of input variables and satisfy a coin theorem and level $1$
bound of $B(n)$, but have $\levelk[3](\F) \geq \Omega(n\cdot B(n))$,
thereby showing that these properties themselves are not enough to give
a positive answer to \cref{conj:meta}.
\begin{prop}\label{prop:highkcounterexample}
	For every function $\sqrt{\log n} + O(1) < B(n) < \sqrt{n}$ and
	sufficiently large odd integer $n$, there is a class $\F_B$ of Boolean
	functions on at most $n$ bits that is closed under restriction,
	negation of input variables, and negation of outputs such that
	$\levelk[1](\F) \leq B$ and $\F$ satisfies a coin theorem of $\abseps
	\cdot O(B)$ for all $\eps = O(1/B)$, but $\levelk[3](\F) = \Omega(B
	\cdot n)$.
\end{prop}

The idea is that it is easy to construct such a family $\F$ if we
consider functions $f:\pmone^n\to [-1, 1]$ which are not Boolean but
instead take on values inside the interval: in particular, if $f$ takes
on the values $\set{\pm B(n)/\sqrt{n}}$, then $\levelk[1](f)$ will be
bounded by $B(n)$ and so $f$ will satisfy a coin theorem of bound
$\abseps \cdot O(B(n))$, but $f$ need not satisfy higher $\levelk$
bounds of the form $O(B(n))^k$. By randomly rounding such a function $f$
to have Boolean outputs, the resulting function (and its closure under
restriction and negation of input variables) will with high probability
still satisfy a coin theorem but will keep its large $\levelk$ mass.
Thus, this suggests that any positive answer to \cref{conj:meta} will
require a condition which is in some sense ``not linear'' and can detect
such bad examples. Formally, \cref{prop:highkcounterexample} follows
from the following lemma which shows that the above process indeed
preserves the sums of Fourier coefficients with high probability.

\begin{lem}\label{lem:roudingfouriersumconcentration}
	Let $m$ be a positive integer and $g:\pmone^m \to [-1, 1]$ be a
	function. Then defining $\randg:\pmone^m \to \pmone$
	as the random function which for each $x\in\pmone^m$ independently
	sets $\randg(x)\in \pmone$ to have expectation $g(x)$, it holds for
	every collection $\mathcal T \subseteq 2^{[m]}$ of subsets of $[m]$
	that
	\[
		\PR{\abs{\sum_{S\in \mathcal T}\fourier{\randg}{S} - \sum_{S\in \mathcal T}\fourier{g}{S}}\geq\eps}
		\leq
		2\exp\of{-2^{m - 1}\eps^2/\abs{\mathcal T}}
	\]
\end{lem}
\begin{proof}
	We first write
	\begin{align*}
		\sum_{S\in \mathcal T}\fourier{\randg}{S}
		- \sum_{S\in \mathcal T}\fourier{g}{S}
		&=
		2^{-m} \sum_{x\in \pmone^m} \of{\randg(x)
		-  g(x)}\cdot\of{\sum_{S\in \mathcal T}\prod_{i\in S} x_i}.
	\end{align*}
	by definition of Fourier coefficients. Then by definition of $\randg$,
	the right hand side is a sum of $2^m$ independent mean-zero random
	variables, one for each $x\in \pmone^n$, bounded in a range of size
	$2^{-m}\cdot 2 \cdot \abs{\sum_{S\in \mathcal T}\prod_{i\in S} x_i}$,
	where the last term depends on $x$. Thus, Hoeffding's inequality
	\cite{hoeffding_probability_1963} implies the concentration bound
	\begin{align*}
		\PR{\abs{\sum_{S\in \mathcal T}\fourier{\randg}{S} - \sum_{S\in \mathcal T}\fourier{g}{S}}\geq\eps}
		&\leq
		2\exp\of{-\frac{2 \cdot \eps^2}{\sum_{x\in \pmone^m}\of{2 \cdot 2^{-m}\cdot\abs{\sum_{S\in \mathcal T}\prod_{i\in S}x_i}}^2}}\\
		&=
		2\exp\of{-\frac{2^{m - 1}\cdot \eps^2}{2^{-m}\sum_{x\in \pmone^m} \of{\sum_{S\in \mathcal T}\prod_{i\in S}x_i}^2}}
	\end{align*}
	It remains to show that the denominator is equal to $\abs{\mathcal
	T}$: first note that this sum can be written as $\E[x\sim
	{\coinset0m}]{\of{\sum_{S\in \mathcal T}\prod_{i\in S}x_i}^2}$ where
	the $x_i$ are distributed as iid random signs. Note that for $S\neq
	\emptyset$ we have that $\prod_{i\in S}x_i$ has mean zero and is
	marginally distributed as a random sign, and for $S = \emptyset$ we
	have that $\prod_{i\in S}x_i = 1$, so that if $\emptyset\in \mathcal
	T$ we can write $\E[x\sim {\coinset0m}]{\of{\sum_{S\in \mathcal
	T}\prod_{i\in S}x_i}^2} = 1 + \E[x\sim {\coinset0m}]{\of{\sum_{S\in
	\mathcal T\setminus \set{\emptyset}}\prod_{i\in S}x_i}^2}$ with
	$\abs{\mathcal T\setminus \set{\emptyset}} = \abs{\mathcal T}-1$, and
	thus it suffices to consider the case that $\mathcal T$ does not
	contain the empty set.

	In this case, we have that $\E[x\sim {\coins[0]}]{\of{\sum_{S\in
	\mathcal T}\prod_{i\in S}x_i}^2}$ is the variance of a sum of
	$\abs{\mathcal T}$ random variables each marginally distributed as a
	random sign. Since all the terms are distinct, given $S\neq S'\in
	\mathcal T$ we know there exists some $j$ in the symmetric difference
	of $S$ and $S'$ (without loss of generality in $S$), and thus the
	covariance of $\prod_{i\in S}x_i$ and $\prod_{i\in S'}x_i$ is zero, as
	$x_j$ has mean zero even conditioned on the value of $\prod_{i\in
	S'}x_i$.  Hence, these variables are all uncorrelated, and so the
	variance of their sum is simply the sum of the variances, which
	is $\abs{\mathcal T}\cdot 1 = \abs{\mathcal T}$ as desired.
\end{proof}

Applying \cref{lem:roudingfouriersumconcentration} to the majority function
proves \cref{prop:highkcounterexample}.
\begin{proof}[Proof of \cref{prop:highkcounterexample}]
	Let $f:\pmone^n \to \pmone$ be the majority function on $n$ bits for
	$n$ odd, and
	$\randf:\pmone^n\to \pmone$ be the random function which independently
	for each $x\in \pmone^n$ sets $\randf(x)$  with expectation $(B/\sqrt n)
	\cdot f(x)$. Then define $\randF$ to be the set of all restrictions of
	all functions of the form $z \mapsto \tau \cdot \randf(z_1\sigma_1,
	\dots, z_n\sigma_n)$ for $\tau \in \pmone$ and $\sigma \in \pmone^n$.

	We claim that with positive probability $\randF$ has the desired
	properties. Note that $\randF$ is closed under restriction, negation
	of input variables, and negation of the output by definition.  To
	prove the coin theorem claim, by Tal's result
	(\cref{lem:restrictfl1givescoin}) it is enough to prove the level $1$
	bound $\levelk[1](\randF) \leq B$.  Thus, we need to show that with
	positive probability $\levelk[1](\randF) \leq B$ and
	$\levelk[3](\randF) \geq \Omega(B\cdot n)$ (we will in fact show an
	upper bound of $B + 1=O(B)$, which is equivalent by shifting).

	Since $\levelk[3](f) = \Theta(n^{3/2})$ (see e.g.~\cite[Problem
	5.26]{odonnell_analysis_2014}), we have $\levelk[3]\of[\big]{(B/\sqrt n) \cdot f}
	= \Theta(B\cdot n)$, so applying
	\cref{lem:roudingfouriersumconcentration} to $(B/\sqrt n) \cdot f$ and
	$\mathcal T = \set{S\subseteq [n]\given \abs S = 3}$ implies that
	\[\PR{\levelk[3](\randF) \geq \levelk[3](\randf) \geq
	\levelk[3](B/\sqrt n \cdot f) - n = \Omega(B\cdot n)} \geq 1 -
	\exp\of{-\Omega(2^n/n^3)}.\]

	\newcommand{\restf}{\randf|_{\rho}^{\sigma,\tau}}
	For the level one bound, since $\randF$ is closed under negation of
	input variables, it suffices to bound the sum of the level $1$ Fourier
	coefficients of each member of $\randF$. For this, consider a fixed
	sign $\tau\in \pmone$, sign pattern $\sigma \in \pmone^n$, and
	restriction $\rho$ of $z \mapsto \tau \cdot
	\randf(z_1\sigma_1, \dots, z_n\sigma_n)$ keeping $m$ variables
	alive, which we denote $\restf$. By Parseval and Cauchy-Schwarz, the sum of the level $1$
	Fourier coefficients of any Boolean function on $m$ variables is at
	most $\sqrt m$, so if $\sqrt m \leq B$ then
	$\levelk[1](\restf) \leq B$ with probability $1$. Thus,
	assume without loss of generality that $m > B^2$.
	Define $g:\pmone^m \to \set{\pm B/\sqrt{n}}$ be given by
	$g(z_1, \dots, z_m) = (B/\sqrt{n})\cdot \tau \cdot f(x_1, \dots, x_n)$
	where $x_i$ is equal to $\sigma_i\cdot \rho(i)$ if $i$ is fixed by $R$, and equal to
	$\sigma_j \cdot z_j$ for $j$ the index of $i$ in the free coordinates
	otherwise. Then the sum of the singleton Fourier coefficients of the
	restriction $\restf$ is distributed exactly as the sum of the
	Fourier coefficients $\randg:\pmone^m \to \pmone$ where $\randg(x)$ is
	set independently for each $x\in \pmone^m$ with expectation $g(x)$.
	Since the sum of $g$'s singleton Fourier coefficients is at most
	$(B/\sqrt n) \cdot \sqrt m \leq B$, by
	\cref{lem:roudingfouriersumconcentration} we have
	\[
		\PR{\levelk[1](\restf) \geq B + 1}
		\leq 2\exp\of{-2^{m - 1}/m}.
	\]
	Now, by assumption $m \geq B^2 \geq \log n + O(\sqrt{\log n})$, so
	since $m \mapsto 2^m/m$ is increasing in $m$ for $m\geq 1/\ln 2$, this
	probability is at most $\exp(-\omega(n))$.  Thus, since there are
	$2^{O(n)}$ functions $\restf$, we conclude by a union bound.
	\let\restf\undefined
\end{proof}

It remains to give the proof of \cref{lem:restrictfl1givescoin}.  For
intuition, writing $g(\eps)=\E{f\of{\coins}}$, recall that in
\cref{prop:fourier_sum_bound} we argued that for arbitrary $f$ and small
$\eps$, the error of the linear approximation $g(\eps)\approx
g(0)+\eps\cdot g'(0)$ is small, where $g'(0)=\sum_{i=1}^n\fourierset fi$.
By contrast, the proof of \cref{lem:restrictfl1givescoin} bounds
$g'(x)$ for small $x$ in terms of the level $1$ Fourier coefficients
of specific \emph{restrictions} of $f$.

\begin{proof}[Proof of \cref{lem:restrictfl1givescoin} \cite{tal_private_2019}]
	Letting $g(\eps)=\E{f\of{\coins}}$, Tal \cite{tal_private_2019}
	observed that for every $b\in[-1,1]$, the derivative $g'(b)$ can be
	bounded in terms of the level $1$ Fourier coefficients of restrictions of $f$
	via a technique of Chattopadhyay, Hatami, Hosseini, and Lovett
	\cite[Proof of Claim 3.3]{cha_hat_hos_lov_pseudorandom_2019}.
	Formally, for $b\in [-1, 1]$, let $\dist{b}$ be the distribution over
	restrictions $\rho$ where each coordinate is independently fixed to
	$\operatorname{sign}(b)$ with probability $\abs b$ and left free with
	probability $1-\abs b$. Then for every $\eps \in [\abs{b} - 1, 1 -
	\abs{b}]$, we have
	\[
		\E{f\of{\coins[b+\eps]}}
		=
		\E[\rho\sim \dist{b}]{f|_\rho\of{\coins[\eps/(1-\abs{b})]}}.
	\]
	Thus, writing $g|_\rho(\eps)=\E{f|_\rho\of{\coins}}$, we have
	by linearity of the derivative and \cref{prop:fourier_sum_bound} that
	\[ g'(b)=\E[\rho\sim\dist{b}]{\frac1{1-\abs b}\cdot g|_\rho'(0)}
		=\frac1{1-\abs b}\E[\rho\sim\dist{b}]{\sum_{i=1}^n \fourierset{f|_\rho}{i}}\,.\]
	In particular, by our assumption on the level $1$ Fourier coefficients of restrictions
	$f|_\rho$, we get that $\abs{g'(b)}\leq t/(1-\abs b)$ for all $b\in[0,1)$, and thus
	the claim follows from the fact that $\fdiff=\int_0^\eps g'(b)\,db$.
\end{proof}

\section{Acknowledgements}

We thank Salil Vadhan for helpful conversations and his feedback on this
note; Avishay Tal for his feedback on this note, telling us about
\cref{lem:restrictfl1givescoin}, and allowing us to include it and its
proof in this note; and the anonymous reviewers for their helpful
comments and suggestions.

\bibliographystyle{tocplain}   
\bibliography{main}

\begin{cjtcsauthors}
\begin{authorinfo}[rohit]
 Rohit Agrawal\\
 Harvard University, USA\\
 rohitagr\tocat{}seas\tocdot{}harvard\tocdot{}edu \\
 \url{https://rohitagr.com}
\end{authorinfo}
\end{cjtcsauthors}

\end{document}